\documentclass[a4paper,12pt,reqno,]{scrartcl}
\usepackage[latin1]{inputenc}
\usepackage{amsmath, amsthm, amssymb}

\newtheorem{satz}{Theorem}[section]
\newtheorem{hsatz}[satz]{Proposition}
\newtheorem{lem}[satz]{Lemma}

\newtheorem{defn}[satz]{Definition}

\theoremstyle{remark}
\newtheorem{bem}[satz]{Remark}
\newtheorem{bsp}[satz]{Example}

\theoremstyle{definition}

\newcommand{\kz}{\mathbb C}
\newcommand{\rz}{\mathbb R}
\newcommand{\nz}{\mathbb N}

\newcommand{\K}{\mathbb K}
\newcommand{\fa}{\text{ } \forall \text{ }}

\newcommand{\dom}{\operatorname{D}}
\newcommand{\hil}{\mathcal H}
\newcommand{\gil}{\mathcal G}

\newcommand{\with}{\text{ with }}
\newcommand{\esgibt}{\exists \text{ }}
\newcommand{\ep}{\varepsilon}
\newcommand{\form}{\mathfrak h} 

\newcommand{\tr}[1]{\operatorname{tr}_{ #1 }} 
\newcommand{\str}[1]{\operatorname{str}_{#1}} 
\newcommand{\sgn}{\operatorname{sgn}} 

\newcommand{\komp}{_{\mathrm{komp}}}
\newcommand{\fin}{_{\mathrm{fin}}}

\newcommand{\ukl}{u}
\newcommand{\rrv}{\ell^2(E_v;\kz)}  

\usepackage[style=alphabetic,backend=bibtex8,maxnames=5]{biblatex}
\usepackage[style=alphabetic]{biblatex}
\usepackage[]{csquotes}
\bibliography{lit.bib}

\begin{document}
\newcommand{\dx}{\,dx}

\title{Unbounded quantum graphs with unbounded boundary conditions}
\author{Daniel Lenz, Carsten Schubert and Ivan Veseli\'c}

\maketitle

\begin{abstract}
We consider metric graphs with a uniform lower bound on the edge lengths but no further restrictions. 
We discuss how to describe  every local self-adjoint Laplace operator on such graphs  by boundary conditions in the vertices given  by  projections and  self-adjoint operators.  We then  characterize  the  lower bounded self-adjoint Laplacians   and determine their associated quadratic form in terms of the operator families encoding the boundary conditions. 

\vspace{8pt}

\vspace{2pt}

\end{abstract}

\section*{Introduction} 
Quantum graphs, i.\,e. metric graphs together with a differential operator,  have attracted a lot of attention in recent years (see e.\,g. the conference proceedings \cite{BCFK,AGA} and the references therein).  On the one hand this is due to their relevance in physics as models for nanostructures. On the other hand this is due to their mathematical features stemming from their role as an intermediate structure between discrete and continuum models.

The basis for the  investigation of quantum graphs is the definition of a suitable self-adjoint operator on the underlying structure. Accordingly, there has been quite some work devoted to defining such operators.

In their influential paper Kostrykin and Schrader studied Laplacians on a metric star graph with finitely many edges using Lagrangian subspaces \cite{KostrykinS-99b}. This gives all self-adjoint versions of the negative Laplace operator for this specific graph (see work of Harmer \cite{Harmer-00} for related material and work of Carlson \cite{Carlson-98} for earlier consideration in a similar direction as well).
More general graphs were then treated subsequently in various works.
Cheon et al \cite{CheonET} consider boundary conditions from a different perspective. Bruening et al \cite{BrueningGP-08}  deal with quantum graphs in the wider context of the theory of boundary triples (see work of Post \cite{Post} for related material as well). In fact, boundary triplets and Lagrangian subspaces  can be used to describe all self-adjoint Laplacians---even for infinite metric graphs (see appendix or \cite{SSVW}).

An alternative approach to quite general metric graphs is developed by Kuchment in  \cite{Kuchment-04}. This work gives a   description of the self-adjoint Laplacians via some new operators $(P,L)$ encoding the  boundary conditions.
With the help of these new operators the  Laplacian is studied  using its   associated quadratic forms. This   work  is not restricted to star graphs but rather deals with general graphs. Still, it  imposes some restrictions   on the underlying graph structure. In particular, it requires finiteness of all  vertex degrees   as well as a lower bound on the edge lengths. Moreover, the use of form methods means that only operators can be tackled which are bounded below. In fact, the work assumes that all boundary conditions are bounded in the sense that the operators $L$ are bounded. 
To get rid of these restrictions is the starting point of this paper. 

More specifically, we deal with general metric graphs and impose neither a condition of finiteness of vertex degree nor a condition of semiboundedness of the operator. We do, however, keep the assumption of a uniform lower bound on the edge length. In this setting our aim is

\begin{itemize}
\item to give an explicit description of \textit{all} (local) self-adjoint Laplacians on a metric graph via the operators $(P,L)$ used by Kuchment,

\item to  characterize those boundary conditions yielding lower bounded operators. 
\end{itemize}
In order to achieve this goal we will have to deal with  operators  $L$ which are not bounded. 
 
 Along the way we also give a sufficient condition for essential self-adjointness.

\smallskip

The results of this paper have proven  useful in the study of  random Schr\"odinger operator; they are used to obtain   Combes-Thomes estimate and a  geometric resolvent inequality in  \cite{diss_11}.
 The whole paper is based on the dissertation \cite{diss_11} of one of the authors.

\smallskip



The paper is organized as follows: In section \ref{Metric} we introduce metric graphs and discuss the necessary background. Section \ref{Unbounded} then provides the first main result, theorem 2.2,  describing all self-adjoint operators on a given metric graph. Moreover, it contains the result on essential self-adjointness.  A discussion of boundary conditions at a vertex with unbounded degree is given in section \ref{abs_dv_unendl}. The characterization of those boundary conditions giving lower bounded operators can be found in section  \ref{Quadratic}. Some remarks on the requirement of a lower bound on the edge lengths is given in section \ref{ab:KGN}. The appendix provides the connection of our work to the approach via Lagrangian subspaces.

\bigskip

\medskip

\textbf{Acknowledgments.} Financial support from  the German Research Foundation (DFG) is gratefully acknowledged. The authors would also like to thank Christian Seifert  and   Michael Gruber for illuminating discussions.

\section{Metric graphs} \label{Metric}  
The biggest difference between combinatoric and metric graphs is the definition of functions on the graph.
Functions on metric graphs are defined on the edges, which are equipped with a length and can be seen as the corresponding interval.

\begin{defn}
A metric graph is a tuple $\Gamma=(E,V,l,i,j)$ consisting of countable sets of edges $E$ and vertices $V$, a length function $l:E\to (0,\infty]$ giving each edge a length and functions giving each edge a starting point and each finite edge an end point $i:E\to V$, $j:\{ e\in E \with l(e)<\infty \}\to V$.
\end{defn}
Note that we allow for loops and multiple edges in this definition.

The interval $I_e:=(0,l(e))$ will be identified with each edge $e$.
With these intervals we define the space
\begin{equation*}
X_E:= \bigcup\limits_{e \in E}\{ e\}\times I_e
\end{equation*}
and denote functions $f:X_E \to \kz$ by $f_e(t):=f(e,t)$.
The underlying Hilbert space is
\begin{equation*}
L^2(X_E):=
\{f=(f_e)_{e\in E} \with f_e\in L^2(I_e), \sum\limits_{e\in E} \|f_e \|^2_{L^2(I_e)}<\infty  \}
\end{equation*}
with the corresponding Sobolev spaces
\begin{equation*}
W^{1,2} (X_E) := \bigoplus\limits_{e\in E } W^{1,2}(I_e),\qquad W^{2,2}(X_E):= \bigoplus\limits_{e\in E } W^{2,2}(I_e).
\end{equation*}
These spaces are sometimes called decoupled Sobolev spaces, as functions don't need to be continuous in the vertex. We will need this freedom (from continuity assumptions)  to describe all possible boundary conditions.
\begin{defn}
If a vertex $v$ is a starting or end point of an edge $e$, then $v$ and $e$ are called incident. We will denote this relation by $e\sim v$.

\smallskip

Without  loss of generality  we will assume throughout that  there are no isolated vertices and treat connected graphs only.

\smallskip

Let $E_v:=\{(e,0) \with v=i(e) \} \cup \{(e,l(e)) \with v=j(e)  \}$ be the set of outgoing and incoming edges incident to $v$. The degree of a vertex is defined by
\begin{equation*}
d_v:= |\{ (e,0) \with v= i(e)\} \cup \{(e,l(e)) \with v=j(e) \}| =|E_v|.
\end{equation*}
\end{defn}
We do not require finiteness of the $d_v$. If all $d_v$, $v\in V$, are finite, we say that the graph has bounded vertex degree.

From the Sobolev imbedding theorem (e.\,g. theorem 4.12 in \cite{AdamsF}) we know that each function in $W^{j+1,2}(0,l)$ has a representative in $C^j(0,l)$ and can be continuously extended to the boundary. Thus we can define the limits
\begin{align*}
f (0) &:=\lim_{t\to 0} f(t) &      f (l(e)) &:=\lim_{t\to l(e)} f (t) &  &\text{for } f\in W^{1,2}(0,l) \text{ and}\\
f' (0) &:=\lim_{t\to 0} f' (t)  &  f' (l(e))&:=\lim_{t\to l(e)} f' (t) &  &\text{for }f \in W^{2,2}(0,l).
\end{align*}
The following result  gives an imbedding of the boundary values by the Sobolev norm for a given function $f\in W^{1,2}(0,l)$.
\begin{lem}
\label{satz:Sobolev}
For each function $f\in  W^{1,2}(0,l)$ and each $a$ with $0<a\leq l$ the inequality
\begin{align}
\label{Sobolev} &|f(0)|^2 \leq \frac{2}{a} \|f\|^2_{L^2 (0,l)} + a \|f'\|^2_{L^2 (0,l)}
\end{align}
holds.
\end{lem}
The assertion and its proof can be found as lemma 8 in \cite{Kuchment-04}.

This lemma gives the finiteness of many useful quantities. This is particularly useful if the lemma is  applicable uniformly on the whole graph, i.\,e. if there is a uniform lower bound on the edge lengths. If such a  bound is missing there are only few results---see section \ref{ab:KGN}.  

\smallskip

For this reason we will (mostly) assume that our graphs are of   of bounded geometry i.\,e. satisfy
\begin{equation}
\label{geom:u}\tag{geom:\ensuremath{\ukl}}
\esgibt \ukl>0 \text{ with } l(e) \geq \ukl \text{ for all } e\in E.
\end{equation}

\begin{defn}
By $\tr{}(f)$ we define the trace of a function $f\in W^{1,2}(X_E)$ to be the vector of all boundary values of $f$ and $\tr{v}(f)$ its restriction to all beginnings/ends of edges incident to $v$:
\begin{equation*}
\tr{}(f)=\left(\left( (f_e(t)\right)_{(e,t)\in E_v}\right)_{v\in V}, \qquad \tr{v}(f):=(f_e(t))_{(e,t)\in E_v}.
\end{equation*}
Analogue we define the signed trace
\begin{equation*}
\str{}(f)=\left(\left( (\sgn(e,t)\,f_e(t)\right)_{(e,t)\in E_v}\right)_{v\in V}, \qquad \str{v}(f):=\left(\sgn(e,t)\,f_e(t)\right)_{(e,t)\in E_v},
\end{equation*}
where $\sgn(e,t)=1$ for $t=0$ and $\sgn(e,t)=-1$ for $t=l(e)$.
\end{defn}
If we look at $\str{}(f')$ the minus sign at the derivatives of the end points give the so called ingoing derivatives (where ingoing refers to the edges). Hence if we change the direction of one (or all) edges, the vector $\str{}(f)$ stays the same.

If there exists a minimum edge length $u_v$ for all edges incident to $v$, \eqref{Sobolev} gives a bound of the $\ell^2$-norm of $\tr{v}(f)$ by the Sobolev-norm of $f$ which yields in particular $\tr{v}(f),\str{v}(f')\in \rrv$. The same holds for the whole graph.

For the study of minimal and maximal operators and their extensions resp. restrictions we define the Sobolev spaces with vanishing boundary values. The space $W^{k,p}_0(X_E)$ is the closure of $C^{\infty}_0(X_E)=\left(\prod\limits_{e\in E}C^\infty_0(I_e)\right) \cap W^{k,p}(X_E)$ in $W^{k,p}(X_E)$. Thus we get by Sobolev embedding (see e.\,g.  lemma \ref{satz:Sobolev})
\begin{align*}
W^{1,2}_0(X_E)&=\{f\in W^{1,2}(X_E) \with \tr{v}(f) =0 \text{ for all }v\in V\},\\
W^{2,2}_0(X_E)&=\{f\in W^{2,2}(X_E) \with \tr{v}(f)=\str{v}(f')=0 \text{ for all }v\in V\}.
\end{align*}

The Laplace operator $\Delta$ defined on $W^{2,2}_0(X_E) $ is not self-adjoint, but symmetric and its adjoint is defined on $W^{2,2}(X_E)$.
With the help of the boundary vectors $\tr{v}(f)$ and $\str{v}(f')$ we can define boundary conditions to make the Laplace operator self-adjoint.

\begin{defn}
A metric graph $\Gamma$ with a self-adjoint differential operator is called Quantum graph.
\end{defn}

\begin{defn}
Let $\Gamma$ be a metric graph with a uniform bound of the edge lengths from below. A boundary condition \eqref{RB:PL} consists of a pair $(P,L)$ satisfying the following conditions:  $P=(P_v)_{v\in V}$ is a family of orthogonal projections $P_v:\rrv \longrightarrow \rrv$ on closed subspaces of $\rrv$ and $L= ((L_v,\dom(L_v)))_{v\in V}$ a family of self-adjoint operators
\begin{align*}
L_v :\dom(L_v)  \longrightarrow (1-P_v)\left(\rrv\right) \qquad \text{with }\dom(L_v) \subset (1-P_v)\left(\rrv\right).
\end{align*}

The negative Laplacian with boundary conditions of the form \eqref{RB:PL} is defined by:
\begin{align}
\notag &\hspace{1.1em} H^{P,L}f=-f'',\\
\tag{BC:P,L}&\begin{aligned}
\dom(H^{P,L})&=\{ f\in W^{2,2}(X_E) \with \forall\  v\in V: \tr{v}(f)\in\dom(L_v),\\
& \hspace{6.2cm} L_v \tr{v}(f)=(1-P_v)\str{v}(f') \}.
\end{aligned}
\end{align}
\end{defn}

A few comments on the definition are in order: 
As the first part of the two restrictions in the boundary condition gives automatically $(1-P_v)\tr{v}(f) = \tr{v}(f) $, we will read the restrictions successively and omit the projection in the second part. Thus we will consequently write $L_v \tr{v}(f)$ instead of $L_v(1-P_v)\tr{v}(f)$ which is commonly used in the literature.

The positive part $L_v^+$ of $L_v$ is defined by
\begin{equation*}
L_v^+x:=L_v P_{[0,\infty)}(L_v) x, \qquad \text{for all  $ x$ with $P_{[0,\infty)} x \in \dom(L_v)$},
\end{equation*}
where $P_{[0,\infty)}$ denotes the spectral projection on the interval $[0,\infty)$.
Analogously we define the negative part $L_v^-$.
Then $L_v^+$ and $-L_v^-$ are non-negative self-adjoint operators acting from $\dom(L_v)$ to $\dom(L_v)$ with the following decomposition:
\begin{equation*}
L_vx=L_v^+x+L_v^-x \text{ for all }x\in \dom(L_v).
\end{equation*}

In our later discussion of lower bounds for $H^{P,L}$ we will need the following property of the family $(L_v)$:
\begin{equation}
\label{S}\tag{BC:\ensuremath{S}}
\esgibt S>0 \text{ with } \langle L^-_v x,x\rangle \geq -S\langle x,x\rangle \text{ for all } x \in \dom(L_v) \text{ and }v\in V.
\end{equation}
\begin{defn}
Let $\Gamma$ be a metric graph and boundary conditions of the form \eqref{RB:PL} be given, which suffice \eqref{S}.
We will call those boundary conditions of the form \eqref{RB:PLS}.
\end{defn}

\begin{bem}
There are different versions of parametrization of boundary conditions and associated properties of metric graphs in the literature
\begin{enumerate}
\item
Kostrykin and Schrader proved in 1999 in \cite{KostrykinS-99b}, that on a star graph with $n$ infinite edges
given two $n\times n$ matrices $A$ and $B$ with $\operatorname{rank}(A,B)=n$ the negative Laplace-operator
$H^{(A,B)}$ with domain $D(H^{(A,B)})=\{f\in W^{2,2}(X_E) \with A\tr{v}(f)+B\str{v}(f')=0  \} $ and $H^{(A,B)}f=-f'' $
is self-adjoint iff $AB^*$ is self-adjoint.
This parametrization coincides with Lagrangian subspaces.
\item
Kuchment proved 2004 in \cite{Kuchment-04} that the boundary conditions from 1. for finite vertex degree can be rewritten in the form \eqref{RB:PL}.
Here $P $ denotes the orthogonal projection on $\ker B$ and $L:(1-P)\left(\kz^{E_v}\right)\to (1-P)\left(\kz^{E_v}\right) $ with $L=(Q^*B(1-P))^{-1}A (1-P)^*$ is a self-adjoint operator, where $Q$ is the orthogonal projection on the image of $B$.

In this work it is also  shown, that the operator
\begin{align*}
H^{P,L}f&=-f''\\
\dom(H^{P,L})&=\{ f\in W^{2,2}(X_E) \with \forall \ v\in V: P_v\tr{v}(f)=0, \\
&\hspace{4.5cm} L_v\tr{v}(f)+(1-P_v)\str{v}(f')=0\}
\end{align*}is self-adjoint, if the metric graph $\Gamma$ and the boundary conditions \eqref{RB:PL} satisfy the following three conditions:
\begin{enumerate}
\item
The edge lengths are uniformly bounded from below
${l(e)\geq \ukl>0}$ for all $e\in E$,
\item the vertex degree is bounded: $d_v<\infty$  for all $v\in V$,
\item and
the norms of the operators $L_v$ from the boundary condition are uniformly bounded by  $\|L_v\| \leq S$.
\end{enumerate}

We want to comment that our definition of the operator $L_v$ is equivalent to $-L_v$ in Kuchment's works, which we find more convenient (see e.\,g. theorem \ref{satz_2}).
\item
Let $\Gamma$ be a metric graph with \eqref{geom:u}. Then we conclude from the Sobolev theorem for $h\in W^{1,2}(X_E)$ and each $\ep$ with $0<\ep\leq \ukl$:
\begin{equation}
\label{gl_Sobolev_Kante}
|h_e(0)|^2 \leq \frac{2}{\ep} \|h_e\|^2_{L^2 (I_e)} + \ep \|h'_e\|^2_{L^2 (I_e)}.
\end{equation}
Summing over all beginnings and ends of edges $e$ with $(e,t)\in E_v$ we get
\begin{equation}
\label{gl_Sobolev_Knoten}
|\tr{v}(h) |^2=
\sum\limits_{(e,t)\in E_v} |h_e(t)|^2
\leq 2 \cdot \sum\limits_{e \sim v}\left( \frac{2}{\ep} \|h_e\|^2_{L^2 (I_e)} + \ep \|h_e'\|^2_{L^2 (I_e)}\right).
\end{equation}
Summing over all vertices results in:
\begin{equation}
\label{gl_Sobolev_Graph}
\sum\limits_{v\in V} |\tr{v}(h)|^2=
\sum\limits_{v \in V} \sum\limits_{(e,t)\in E_v} |h_e(t)|^2
\leq 2 \cdot \sum\limits_{e \in E} \left( \frac{2}{\ep} \|h_e\|^2_{L^2 (I_e)} + \ep \|h_e'\|^2_{L^2 (I_e)}\right).
\end{equation}
\end{enumerate}
\end{bem}

\begin{bem}
\label{bem_rrv}
Let $\Gamma$ be a metric graph with a lower bound of the edge lengths in each vertex separately, i.\,e.
$ \exists \ u_v>0$ with $l(e)\geq u_v $ for all edges $e$ incident to $v$. Then $\rrv$ is the appropriate space, in the sense that
it holds:
\begin{enumerate}
\item[(i)]
 $\tr{v}(f)\in \rrv $  for all $v\in V$ and all $f\in W^{1,2}(X_E)$. (So,
in particular $f\in W^{2,2}(X_E)\Rightarrow \str{v}(f')\in \rrv  $.)
\item[(ii)]
For all $x$, $y \in \rrv$ there exists a function $f\in W^{2,2}(X_E) $ with $\tr{v}(f)=x $, $\str{v}(f')=y $ and $f$ is supported in a small neighborhood containing only one vertex (v). The mapping $(x,y)\mapsto f$, which maps $\rrv \times \rrv$ into $W^{2,2}(X_E)$, is continuous.
\end{enumerate}
\end{bem}

\begin{proof}
\begin{enumerate}
\item[(i)] holds because of \eqref{gl_Sobolev_Knoten}.

\item[(ii)]
For each $x$, $y\in \rrv$ we can construct a function $f$:

For each end $(e,t)\in E_v$ we define the function $f_e$ on an interval with length $\frac{u_v}{2}$ as a polynomial of degree three, s.\,t. at $t$ the boundary value and derivative equals $x_{(e,t)} $ and $y_{(e,t)}$ and on the other side both values vanish. On all other edges and part of edges we continue with zero.

Then the norm $\|f \|_{W^{2,2}(X_E)} $ is bounded by $c(u_v)\,(\|x\|_{\rrv}+\|y\|_{\rrv})$.\qedhere
\end{enumerate}
\end{proof}

\begin{bem}
\label{bem_rrv2}
If we have a metric graph with a uniform lower bound of the edge lengths, $\inf\limits_{e\in E} l(e)\geq \ukl >0$, the same argument as in the last remark yields:
\begin{itemize}
\item[(i)]
The mappings $\tr{},\str{}:W^{1,2}(X_E)\to \bigoplus\limits_{v\in V} \rrv$  are well defined, linear and surjective.
\item[(ii)]
For each $x$, $y\in \bigoplus\limits_{v\in V}\rrv$ there exists a function $f\in W^{2,2}(X_E)$, such that $\tr{}(f)=x$ and $\str{}(f')=y$.
Thus the mapping $f\mapsto (\tr{}(f),\str{}(f'))$ is surjective onto $\bigoplus_{v\in V}\rrv \times  \bigoplus_{v \in V} \rrv $.
Again $(x,y)\mapsto f$ is continuous.
\end{itemize}
\end{bem}

\section{Unbounded boundary conditions} \label{Unbounded}
In this section we show---by direct calculation---that the negative Laplacian with boundary conditions of the form \eqref{RB:PL} on a metric graph with \eqref{geom:u} is self-adjoint.

An important proposition is the following:
\begin{hsatz}
\label{hsatz_0}
Let $\Gamma$ be a metric graph with \eqref{geom:u}.
For all $f$, $g\in W^{2,2}(X_E)$ the equality
\begin{equation*}
\langle f,-g''\rangle-\langle -f'',g \rangle
=\sum\limits_{v\in V} \langle \tr{v}(f),\str{v}(g') \rangle -\sum\limits_{v\in V}\langle \str{v}(f') ,\tr{v}(g) \rangle.
\end{equation*}
holds, where the sums are absolutely convergent.
\end{hsatz}

\begin{proof}
Performing integration by parts twice yields:
\begin{align*}
\langle f,-g'' \rangle_{L^2(X_E)}
&=\sum\limits_{e\in E} -f_e(x)\overline{g_e'(x)}|_0^{l(e)}- \langle f'(x),-g'(x)\rangle \dx\\
&=\sum\limits_{e\in E} -f_e(x)\overline{g_e'(x)}|_0^{l(e)}+ \sum\limits_{e\in E} f'_e(x)\overline{g_e(x)}|_0^{l(e)}
+\langle -f'',g\rangle.
\end{align*}
We want to rearrange the sums
\begin{equation}
\label{gl_rb_sum_abskonv}
\sum\limits_{e\in E} -f_e(x)\overline{g_e'(x)}|_0^{l(e)} +\sum\limits_{e\in E} f'_e(x)\overline{g_e(x)}|_0^{l(e)}
\end{equation}
to
\begin{equation*}
\sum\limits_{v \in V} \langle \tr{v}(f),\str{v}(g') \rangle -\sum\limits_{v \in V}\langle \str{v}(f'),\tr{v}(g) \rangle.
\end{equation*}
The absolute convergence of \eqref{gl_rb_sum_abskonv} can be derived from \eqref{gl_Sobolev_Kante} with $\ep=u$:
\begin{equation*}
|f_e(0)|^2 \leq \left(\frac{2}{u}+u\right)\|f_e \|^2_{W^{1,2}(I_e)}\leq   c \|f_e \|^2_{W^{2,2}(I_e)}.
\end{equation*}
With analog estimates for$|f_e(l(e))|^2 $, $|f'_e(0)|^2  $ and $|f'_e(l(e))|^2 $ we conclude
\begin{align*}
\sum\limits_{e\in E} |f_e(x)\overline{g'_e(x)}|_0^{l(e)} |
&\leq \sum\limits_{e\in E} \left( |f_e(l(e))||\overline{g'_e(l(e))}|+|f_e(0)||\overline{g'_e(0)}|  \right)\\
&\leq 2c\, \sum\limits_{e\in E} \|f_e\|_{W^{2,2}(I_e)}\|g_e\|_{W^{2,2}(I_e)}\\
&=c\,\left( \|f \|^2_{W^{2,2}(X_E)}+\|g\|^2_{W^{2,2}(X_E)}  \right)
<\infty.
\end{align*}
Given the absolute convergence of the sums in question, we can rearrange them according to ingoing derivatives  and in this way derive the desired assertion from \eqref{gl_rb_sum_abskonv}.
\end{proof}


Now we are able to prove our first main result.
\begin{satz}
\label{satz_1}
Let a metric graph $\Gamma$ with $\inf\limits_{e\in E} l(e)\geq \ukl>0$ and boundary conditions of the form \eqref{RB:PL} be given. Then the negative Laplacian
\begin{align*}
\dom(H^{P,L})&=\{ f\in W^{2,2}(X_E) \with \tr{v}(f) \in \dom(L_v), \\
 & \hspace{4.5cm} L_v \tr{v}(f)=(1-P_v) \str{v}(f') \fa v\in V \},\\
H^{P,L}f&=-f''
\end{align*}
is self-adjoint.
\end{satz}

\begin{bem}  Note that $D (L_v)$ is dense in the kernel of $P_v$ by its  very definition. In particular, the following holds:
\begin{itemize}
\item The   requirement $\tr{v}(f)\in \dom(L_v)$ implies   $P_v \tr{v}(f)=0$ (which is the formulation usually discussed in   the literature).
\item If $\|L_v\|<\infty$ one has $P_v(\tr{v}(f))=0 \Leftrightarrow \tr{v}(f)\in \dom(L_v)$.
\end{itemize}
\end{bem}

\begin{proof}
\begin{itemize}
\item
First we show, that $H^{P,L}$ is symmetric. Let $f$, $g\in \dom(H^{P,L})$. From proposition \ref{hsatz_0} we derive:
\begin{equation}
\label{gl_HPL_symmetrisch}
\langle f,H^{P,L}g \rangle_{L^2(X_E)}-\langle H^{P,L}f,g\rangle
=\sum\limits_{v\in V} \langle \tr{v}(f),\str{v}(g') \rangle -\sum\limits_{v\in V}\langle \str{v}(f'),\tr{v}(g) \rangle.
\end{equation}
For $H^{P,L}$ to be symmetric both sums must add to zero, which is in particular satisfied, if
$\langle \tr{v}(f),\str{v}(g')\rangle $ and
$\langle \str{v}(f'),\tr{v}(g)\rangle $ are equal in each vertex.
By assumption we have $\tr{v}(f)=q_1 $ with $q_1\in \dom(L_v) $
and $L_v\tr{v}(f)=(1-P_v)\str{v}(f')$, i.\,e.  $\str{v}(f')=L_vq_1+p_1 $ with $p_1 \in P_v(\rrv)$.
With analog notation for g:  $\tr{v}(g)=q_2$ and $\str{v}(g')=L_vq_2+p_2$ we get:
\begin{align*}
\langle \tr{v}(f),\str{v}(g')\rangle-\langle \str{v}(f'),\tr{v}(g)\rangle&=\langle q_1,Lq_2+p_2 \rangle-\langle L_vq_1+p_1,q_2 \rangle\\
&=\langle q_1,L_v q_2\rangle-\langle L_vq_1,q_2 \rangle \\
&=0
\end{align*}
as $L_v$ is self-adjoint and scalar products of elements in orthogonal subspaces ($\langle p_\cdot ,q_\cdot\rangle$) are zero.

\item
Let $H^*={H^{P,L}}^*$ and $f\in \dom(H^*)$. Then there exists a function $ h\in L^2(X_E)$ with
\begin{equation}
\label{gl_H_stern}
\langle H^{P,L}g,f \rangle =\langle g,h \rangle \fa g\in \dom(H^{P,L}).
\end{equation}

Thus $ f$ lies in $W^{2,2}(X_E)$, since the above equation is satisfied for each test function in $C_0^\infty(I_e)$ (which all lie in the domain of $H^{P,L}$) on each edge and the second weak derivative of $f$ is $-h$.

Thereby we are left to show the boundary condition of $H^{P,L}$ for $f$.
By proposition \ref{hsatz_0} and relation~\eqref{gl_H_stern} we get for all $g\in \dom(H^{P,L}) $
\begin{align*}
\langle g,h \rangle&=\langle H^{P,L}g,f \rangle=\langle-g'',f\rangle \\
&=\sum\limits_{v \in V} \langle \str{v}(g'),\tr{v}(f) \rangle -\sum\limits_{v \in V}\langle \tr{v}(g),\str{v}(f') \rangle+\langle g,-f''\rangle.
\end{align*}
Since $-f''=h $ we obtain for all $g \in \dom(H^{P,L})$
\begin{equation}
\label{gl_rb_terme0}
\sum\limits_{v \in V} \langle \str{v}(g'),\tr{v}(f) \rangle -\sum\limits_{v \in V}\langle \tr{v}(g),\str{v}(f') \rangle=0.
\end{equation}
If we pick functions $g$, with support in a small neighborhood of a vertex $v$ with radius smaller than $\ukl$, then we find
\begin{equation}
\label{gl_gl_kn_rb_terme0}
\langle \str{v}(g'),\tr{v}(f) \rangle =\langle \tr{v}(g),\str{v}(f') \rangle.
\end{equation}
For arbitrary $q_1\in \dom(L_v)$ and $p_1\in P_v(\rrv)$ there is a function $g\in W^{2,2}(X_E)$ with $\tr{v}(g)=q_1$ and $\str{v}(g')=L_vq_1+p_1$ by remark \ref{bem_rrv} (ii).
With relation \eqref{gl_gl_kn_rb_terme0} we get
\begin{align*}
\langle L_vq_1+p_1,\tr{v}(f) \rangle =\langle L_v q_1, \tr{v}(f) \rangle+\langle p_1,\tr{v}(f) \rangle = \langle q_1,\str{v}(f') \rangle.
\end{align*}
Choosing one of $q_1$ and $p_1$ equal to zero and the other arbitrary in the corresponding subspace, we get $P_v \tr{v}(f)=0$ and
$\langle q_1,L_v \tr{v}(f)\rangle=\langle L_v q_1,\tr{v}(f) \rangle = \langle q_1,\str{v}(f')\rangle$,
which yields $\tr{v}(f)\in \dom(L_v^*)=\dom(L_v)$ and $(1-P_v)\str{v}(f')=L_v \tr{v}(f)$.\qedhere
\end{itemize}
\end{proof}

\begin{bem}(Parametrization of self-adjoint Laplacians by vertex boundary conditions of the form \eqref{RB:PL})
The following converse of the statement of the theorem is also true:
For a metric graph $\Gamma$ with $\inf\limits_{e\in E}l(e)\geq \ukl >0$ and a self-adjoint negative Laplacian $H$ with boundary condition acting locally in each vertex, there exists a boundary condition of the form \eqref{RB:PL}, such that the domain of $H$ equals the domain of $H^{P,L}$ as given in the above theorem.
This result can be obtained using the theory of Lagrangian subspaces and boundary triplets and is illustrated in the appendix (see theorem \ref{satz_HG} and theorem \ref{satz_Hsa_LUR}).
In this sense we obtain a parametrization of all self-adjoint Laplacians with local boundary conditions by conditions of the form \eqref{RB:PL}.

The idea to use Lagrangian subspaces to find self-adjoint extensions of the minimal Laplace operator on metric graphs was firstly used in \cite{KostrykinS-99b}.
\end{bem}

\begin{bem}
\label{bem_1} The  proofs of proposition \ref{hsatz_0} and theorem \ref{satz_1} actually show  the following: Let $\Gamma$ be a metric graph with \eqref{geom:u} and $H^{P,L}$ a negative Laplacian with boundary condition of the form \eqref{RB:PL}.
Then the following representation
\begin{align*}
\langle H^{P,L}f,g\rangle
&=\sum\limits_{v\in V} \langle \str{v}(f'),\tr{v}(g)\rangle+\sum\limits_{e\in E}\int\limits_{I_e} f'_e(x)\overline{g'_e(x)}dx\\
&=\sum\limits_{v\in V} \langle L_v \tr{v}(f),\tr{v}(g)\rangle+\langle f',g' \rangle
\end{align*}
holds for all $f\in \dom(H^{P,L})$ and all $g\in \{h\in W^{1,2}(X_E) \text{ with } P_v\tr{v}(h)=0 \fa v\in V\}$.
\end{bem}


\begin{bem}
If $d_v<\infty$ in each vertex, then $L_v$ is bounded and the space $\rrv$ simplifies to
$\kz^{E_v}$.
For $d_v=\infty$ we get some changes compared to the finite case.
Examples and further discussion  follow in section \ref{abs_dv_unendl}.
\end{bem}


We will now turn to   essential  self-adjointness of the operators $H^{P,L}$ for metric graphs with bounded vertex degree. As a preparatory step  we prove the following lemma.

\begin{lem}
\label{hsatz_Cinf_W22}
For each function $f\in W^{2,2}(I)$, with $I$ an interval of the form $(0,l)$, $l\in \rz \cup \{\infty\}$, there exists a smooth function in $C^\infty(I)$ with the same boundary values (function and derivative) and an arbitrary small difference to $f$ in the $W^{2,2}$-norm.
\end{lem}

\begin{proof}
Let $f\in W^{2,2}(0,l)$ and $\ep >0$. For each $\delta >0$ we can find a function $\phi \in C^\infty(0,l) $ with $\|f-\phi\|_{W^{2,2}(0,l)}\leq \delta $ (see e.\,g. theorem 3.17 in \cite{AdamsF}).
Then by \eqref{gl_Sobolev_Kante} we conclude for the boundary values of the difference $f-\phi$
\begin{equation*}
|f(0)-\phi(0)|^2\leq \left( \frac{2}{l} +l \right)\delta^2,\ |f'(0)-\phi'(0)|^2\leq \left( \frac{2}{l} +l \right)\delta^2,\ldots
\end{equation*}
With the notations $a=f(0)-\phi(0)$, $b=f'(0)-\phi'(0) $ and, if $l<\infty$, $c=f(l)-\phi(l)$, $d=f'(l)-\phi'(l)$
we construct the polynomial of degree three with $p(0)=a$, $p'(0)=b$, $p(l)=c $ and $p'(l)=d $:
\begin{equation*}
p(x)=a+bx-\left(\frac{2b+d}{l}+\frac{3(a-c)}{l^2}\right)x^2 + \left( \frac{2(a-c)}{l^3}+\frac{b+d}{l^2} \right)x^3.
\end{equation*}
For an edge with infinite length we chose $c=d=0$ and $l$ can be picked arbitrarily.
We can estimate the Sobolev-norm of $p$ by
\begin{equation*}
\|p(x)\|^2_{W^{2,2}(0,l)}\leq \delta^2 \frac{q(l)}{l^4},
\end{equation*}
where $q$ is a polynomial with non-negative coefficients, which do not depend on $\delta$.
If we pick $\delta$, s.\,t. $\delta + \delta \frac{\sqrt{q(l)}}{l^2}\leq \ep $,we get
\begin{align*}
\|f-(\phi+p)  \|_{W^{2,2}(0,l)} &\leq \|f-\phi \|_{W^{2,2}(0,l)}+\|p\|_{W^{2,2}(0,l)}\\
&\leq \delta +\delta \frac{\sqrt{q(l)}}{l^2}\leq \ep
\end{align*}
and $(\phi-p)$ is the desired function.
\end{proof}

\begin{satz}
Let $\Gamma$ be a metric graph with \eqref{geom:u}, bounded vertex degree and boundary condition of the form \eqref{RB:PL}.
The restriction of $H^{P,L} $ to $C^\infty\fin(P,L)=\{f\in \dom(H^{P,L})\cap C^\infty(X_E)  \text{ with } f_e \not\equiv 0$  only on finitely many edges$\}$ is essentially self-adjoint.
\end{satz}
Here $C^\infty(X_E)$ stands for $=\{ (f_e)_{e\in E} \text{ with } f_e \in C^\infty(0,l(e))  \}$.
\begin{proof}
Let $f\in \dom(H^{P,L})$.
For each $n\in\nz$ pick the index $k(n)\in \nz$ s.\,t.
\begin{equation*}
\sum\limits_{j=1}^{k(n)} \| f_{e_j} \|^2_{W^{2,2}(I_{e_j})} \geq \|f \|^2_{W^{2,2}(X_E)}-\frac{1}{2n}.
\end{equation*}
Then:
\begin{equation}
\label{eq_wes_sa_1}
\sum\limits_{j=k(n)+1}^\infty \| f_{e_j} \|^2_{L^{2}(I_{e_j})}, \sum\limits_{j=k(n)+1}^\infty \| f'_{e_j} \|^2_{L^{2}(I_{e_j})},
\sum\limits_{j=k(n)+1}^\infty \| f''_{e_j} \|^2_{L^{2}(I_{e_j})} \leq \frac{1}{2n}.
\end{equation}

Let $V_{n}:=\left\{i(e_m)\text{ with } m\in \{1,2,\ldots k(n)\} \right\}\cup
\left\{j(e_m)\text{ with } m\in \dom(j)\cap \{1,2,\ldots k(n)\} \right\}$ and
$\psi(x) $ a smooth function mapping the interval $(0,\ukl)$ onto $[0,1]$, which is identically one in a neighborhood of zero and identically zero in a neighborhood $\ukl$.
Then we construct cut-off functions
\begin{align*}
\psi_n(x)=
\begin{cases}
  \begin{cases}
    \psi(x)  \text{ on }(0,\ukl)\\
    \equiv 0 \text{ on } [\ukl,l(e))
  \end{cases} & \text{on edges }e \text{ with } i(e)\in V_n \wedge j(e)\not\in V_n\\
  \begin{cases}
    \psi(l(e)-x)  \text{ on }(l(e)-\ukl,l(e))\\
    \equiv 0 \text{ on } (0,l(e)-\ukl]
  \end{cases} & \text{on edges }e \text{ with } j(e)\in V_n \wedge i(e)\not\in V_n\\
  \equiv 0 & \text{on edges }e \text{ with } i(e) \not\in V_n \wedge j(e)\not\in V_n\\
  \equiv 0 & \text{on edges }e \text{ with } i(e) \not\in V_n \wedge l(e)=\infty\\
  \equiv 1 & \text{on all remaining edges.}
\end{cases}
\end{align*}
From lemma \ref{hsatz_Cinf_W22} we take approximations $\varphi_{n,e}$ of $f_e$, s.\,t.
$\|f_e-\varphi_{n,e} \|^2_{W^{2,2}(I_e)}\leq \frac{1}{2n\cdot k(n)}$.
Then the sequence $f_n:=\varphi_n \psi_n $ consists of functions in $ C^\infty\fin(P,L) $, which by \eqref{eq_wes_sa_1} satisfy
\begin{equation*}
\|f_n-f \|^2_{L^{2}(X_E)}\leq \frac{1}{n}.
\end{equation*}
Thus $f_n \xrightarrow{L^2(X_E)} f$ and analogously $-\varphi_n'' \psi_n \xrightarrow{L^2(X_E)}-f'' $.

The function $f_n$ satisfies the same boundary condition as $\varphi_n$. With the uniform bound $c=\frac{2}{\ukl^2}$ of $\psi'$ and $\psi''$ we get with \eqref{eq_wes_sa_1}
\begin{align*}
\|\varphi'_n \psi_n'\|^2_{L^2(X_E)} &\leq \frac{c}{2n},\\
\|\varphi_n \psi_n''\|^2_{L^2(X_E)} &\leq \frac{c}{2n}
\end{align*}
and altogether
\begin{equation*}
H^{P,L}f_n=-f_n''=-\varphi_n\psi_n''-2\varphi_n'\psi'_n-\varphi_n''\psi_n \to -f''=H^{P,L}f. \qedhere
\end{equation*}
\end{proof}


\section{Boundary conditions at unbounded vertex degree}
\label{abs_dv_unendl}

In this section we give examples of operators on metric graphs with unbounded vertex degree.

Obviously Dirichlet ($P_v=Id$ and $L_v=0$) and Neumann ( $P_v=0$ and $L_v=0$) boundary conditions give self-adjoint operators as the boundary condition decouple the edges. In this two cases it is irrelevant if there are finitely ore infinitely many incident edges.

Other boundary conditions change dramatically the domain and properties of the corresponding Laplacian under infinite vertex degree: We will discuss this in various examples.

\medskip

\begin{bsp}
Functions in the domain of the Laplacian with $\delta$-type boundary conditions are continuous in the vertex and the sum of the ingoing derivatives is equal to $\alpha_v$ times the value of the function in the vertex with a real parameter $\alpha_v$:
\begin{equation*}
\sum\limits_{(e,t)\in E_v}\sgn(e,t)\cdot f'_e(t)=\alpha_v \cdot f_e(t) \qquad \text{for all }(e,t) \in E_v.
\end{equation*}
The special case $\alpha_v=0$ is called Kirchhoff or free boundary condition. We will now formalize this via the $(P,L)$  approach. Two cases have to be distinguished:

\begin{itemize}
 \item
For a vertex with finite degree the projection $P_v$ and the operator $L_v$ can be represented as
\begin{align*}
P_v=\frac{1}{d_v}\begin{pmatrix} d_v-1 & -1 & -1 & \ldots &  -1 \\ -1 & d_v-1 & -1 & \ldots & -1 \\-1 & -1 & d_v-1 & \ldots & -1 \\
\vdots &&& & \vdots \\-1 & -1 & -1 & \ldots & d_v-1  \end{pmatrix},
\qquad L_v=\frac{\alpha_v}{d_v}.
\end{align*}
\item
For a vertex with $d_v=\infty$ the continuity of the function gives $\tr{v}(f)\equiv 0$ (as $\tr{v}(f)\in \rrv$).
Therefore all $\delta$-type boundary conditions in such a vertex  are equal to  Kirchhoff boundary condition.
The attempt to construct a parameterization of the form \eqref{RB:PL} then  results in ${1-P_v=0}$. As will be discussed in the next theorem, this yields Dirichlet boundary condition.
\end{itemize}
\end{bsp}


\begin{satz}
\label{satz_kichhoff_unendlich}
Let $\Gamma$ be a metric graph with \eqref{geom:u} and at least one vertex $v_0\in V$ with infinite vertex degree. The negative Laplace operator $H_\mathrm{K}$ with Kichhoff boundary condition in $v_0$, and $P_v$, $L_v$ of the form \eqref{RB:PL} in all other vertices, is symmetric, but not self-adjoint. The operator is not closed and its closure has a Dirichlet boundary condition at $v_0$.
\end{satz}

\begin{proof}
The domain of the operator is
\begin{align*}
&\dom(H_\mathrm{K})=\{ W^{2,2}(X_E) \text{ with } \  \tr{v_0}(f) \equiv 0 \text{ and } \sum\limits_{(e,t) \in E_{v_0}}\sgn(e,t)\, f'_e(t)=0,\\
&\hspace{2.5cm}\tr{v}(f)\in \dom(L_v),\ (1-P_v)\str{v}(f')=L_v\tr{v}(f) \text{ in all vertices }v\neq v_0\}.
\end{align*}
By \eqref{gl_HPL_symmetrisch} $H_\mathrm{K}$ is symmetric if
\begin{equation*}
\sum\limits_{v\in V} \langle \tr{v}(f),\str{v}(g') \rangle -\sum\limits_{v\in V}\langle \str{v}(f'),\tr{v}(g) \rangle=0.
\end{equation*}
This is true, since in all vertices $v\neq v_0$ we have chosen boundary conditions of the form \eqref{RB:PL} and in $v_0$ we have $\tr{v_0}(f)=\tr{v_0}(g)\equiv 0$.

Let $f\in \dom({H_\mathrm{K}}^*)$. Then we conclude---analog to the proof of theorem \ref{satz_1}---that $f\in W^{2,2}(X_E) $,  ${H_\mathrm{K}}^* f=-f'' $ and
equation \eqref{gl_rb_terme0} is satisfied  for all $g\in \dom(H_\mathrm{K})$.
In the vertex $v_0$ with  $\tr{v_0}(g)\equiv 0$ we conclude from \eqref{gl_gl_kn_rb_terme0}:
\begin{equation*}
 \langle  \str{v_0}(g'),\tr{v_0}(f) \rangle=0 \qquad \text{for all }g\in\dom(H_\mathrm{K}).
\end{equation*}

Using the vectors $(-1,1,0,0,\ldots )^T$, $(0,-1,1,0,\ldots )^T$, $(0,0,-1,1,0,\ldots )^T$ and so on for $\str{v_0}(g')$ provides $\tr{v_0}(f)=c\cdot (1,1,1,\ldots)^T$ and thus automatically $\tr{v_0}(f)\equiv 0 $ as $d_v=\infty$.
Hence we get $f\in \dom({H_\mathrm{K}}^*)\Rightarrow f\in W^{2,2}(X_E)\text{ and } \tr{v_0}(f)\equiv 0$.

Conversely, for all functions $f\in W^{2,2}(X_E) $ with $\tr{v_0}(f)\equiv0$, satisfying the given boundary conditions in all other vertices and for all $g\in \dom(H_\mathrm{K})$ we conclude from proposition \ref{hsatz_0}:
\begin{align*}
\langle H_\mathrm{K} g,f \rangle = \langle g,-f'' \rangle +\sum\limits_{v \in V} \langle \str{v}(g'),\tr{v}(f) \rangle -
\sum\limits_{v \in V} \langle \tr{v}(g),\str{v}(f') \rangle.
\end{align*}
As the scalar products in the sums are equal (as in the proof of theorem \ref{satz_1}) this results in:
\begin{equation*}
\langle H_\mathrm{K} g,f \rangle = \langle g,-f'' \rangle \qquad \text{for all }g\in \dom(H_\mathrm{K}).
\end{equation*}
In total we get
\begin{align*}
\dom({H_\mathrm{K}}^*)=\{f\in W^{2,2}(X_E)\text{ with } &\tr{v_0}(f)=0,\ \tr{v}(f)\in\dom(L_v),\\
&L_v\tr{v}(f)=(1-P_v)\str{v}(f') \text{ for all } v\neq v_0  \}.
\end{align*}
The operator ${H_\mathrm{K}}^* $ is equal to the operator $H_\mathrm{D}$ with Dirichlet boundary condition at $v_0$, which is self-adjoint and closed.
As there exists a function in $\dom(H_\mathrm{D})$, which doesn't lie in $\dom(H_\mathrm{K})$ (which follows directly from remark \ref{bem_rrv}) the operator is not self-adjoint.
Now we obtain:
\begin{equation*}
\overline{H_\mathrm{K}}={H_\mathrm{K}}^{**}=\left( H_\mathrm{D} \right)^*=H_\mathrm{D}.\qedhere
\end{equation*}
\end{proof}

\begin{bsp}
The $\delta'$ boundary conditions are defined by
\begin{equation*}
\sum\limits_{(e,t)\in E_v} f_e(t)=\alpha\cdot \sgn(e,t)\cdot f'_e(t) \qquad  \text{for all } (e,t)\in E_v.
\end{equation*}
This case can be treated similarly to the previous one.  In fact,
 in analogy the closure of an operator with $\delta'$-type boundary condition at a vertex with $d_v=\infty$ results in Neumann boundary conditions. 
The proof follows the same arguments as in the case of boundary conditions of $\delta$-type.
Only $\tr{v}$ and $\str{v}$ change roles.
We leave the details to the reader.
\end{bsp}

The previous examples with $d_v = \infty$ resulted in 'trivial' boundary conditions at the vertices with unbounded degree. However, there also  exist  non-trivial boundary conditions for metric graphs with a vertex with infinite vertex degree. This is discussed next.
\begin{bsp}
For a metric graph with \eqref{geom:u} and at least one vertex with $d_v=\infty $
we define the following boundary condition in this vertex:
\begin{align*}
Px&=0 \qquad \text{for all } x\in \ell^2(E_v,\kz),\\
Lx&=
\begin{pmatrix}
\frac{1}{2} & \frac{1}{4} & \frac{1}{8} & \ldots & \frac{1}{2^n} & \ldots \\
\frac{1}{4} \\
\frac{1}{8}\\
\ldots & & & 0\\
\frac{1}{2^n}\\
\ldots
\end{pmatrix}x.
\end{align*}
Then $L$ is bounded, since
\begin{align*}
\|Lx\|_{\ell^2(E_v,\kz)}&=\left(\left( \sum\limits_{k=1}^\infty \frac{1}{2^k}x_k \right)^2
+\sum\limits_{k=2}^\infty \left( \frac{1}{2^k}x_1 \right)^2 \right)^{\frac{1}{2}}\\
&\leq \left( \frac{1}{2}\|x\|^2 +\frac{1}{8}x_1^2 \right)^{\frac{1}{2}}\\
&\leq \|x\|.
\end{align*}
As $L$ is bounded and obviously symmetric it must be self-adjoint.

For arbitrary $x\in \rrv$ and $y:=Lx$ by remark \ref{bem_rrv} we find a function $f \in W^{2,2}(X_E)$ with $\tr{v}(f)=x$ and $\str{v}(f')=y$. This means we find functions $f$ in the domain of the operator with infinitely many non-trivial components in the boundary vectors $\tr{v}(f)$ and $\str{v}(f')$
(e.\,g. for $x_n=\frac{1}{n}$).
\end{bsp} 


\section{Quadratic form and lower bounded operator}\label{Quadratic}
In \cite{Kuchment-04} it was shown, that $H^{P,L}$ is self-adjoint by considering  the associated quadratic form under the restrictions $d_v<\infty $ and a uniform bound of $\|L_v\|$.
In this section we will characterize when $H^{P,L} $ from theorem  \ref{satz_1}
is lower bounded and explicitly compute  the associated quadratic form.  This will in particular show that this  form agrees with the form presented  in \cite{Kuchment-04} (provided the restrictions of that work are imposed).


\begin{satz}
\label{satz_2}
Let $\Gamma$ be a metric Graph with \eqref{geom:u}.
Let $H^{P,L}$ be the negative Laplace operator with boundary conditions of the form \eqref{RB:PL} as in theorem \ref{satz_1}.
The operator $H^{P,L}$ is lower bounded, iff $L_v^-$ is uniformly (lower) bounded, i.\,e. if the boundary conditions have the form \eqref{RB:PLS}.
\end{satz}

\begin{proof}
\begin{itemize}
\item[$\Longrightarrow$:]
Let $H^{P,L}$ be lower bounded. With remark \ref{bem_1} it holds
\begin{align*}
\sum_{v \in V}\langle L_v\tr{v}(f),\tr{v}(f) \rangle_{\rrv} + \|f' \|_{L^2(X_E)}^2=\langle H^{P,L}f,f \rangle_{L^2(X_E)}  \geq -c\,\| f\|_{L^2(X_E)}^2
\end{align*}
for all $f\in \dom(H^{P,L})$, and with the decomposition of $L_v$:
\begin{equation}
\label{gl_Lvp_1}
\sum\limits_{v\in V}\langle L^-_v \tr{v}(f),\tr{v}(f)\rangle_{\rrv}
\geq -c\,\| f\|^2- \|f' \|^2 -\sum\limits_{v\in V}\langle L^+_v \tr{v}(f),\tr{v}(f)\rangle_{\rrv}.
\end{equation}

For some $\mathcal C\geq 0$ we have to show, that $L_v^-$ is uniformly bounded from below, i.\,e.
\begin{equation}
\label{gl_Lvp_2}
\fa v\in V :\qquad \langle L_v^- x,x \rangle_{\rrv} \geq -\mathcal C\, \|x \|_{\rrv} \qquad \fa x \in \dom(L_v).
\end{equation}

By decomposition of $L_v$ in positive and negative part we also get a decomposition of $x$  in $x=x_+ +x_- $ (with $x_\pm=P_\pm x$).
Thus it is sufficient to
prove \eqref{gl_Lvp_2} for all $x_- \in P_-(\dom(L_v))$.

For functions $f\in \dom(H^{P,L})$, supported only in a small neighborhood of one vertex $v\in V$ and $\tr{v}(f)\in P_-(\dom(L_v))$ we get from \eqref{gl_Lvp_1} with $L_v^+\tr{v}(f)=0$
\begin{equation}
\label{gl_Hpl_nubeschr}
\langle L^-_v \tr{v}(f),\tr{v}(f)\rangle  \geq \sum\limits_{e \sim v} \left( -c\,\| f_e\|^2_{L^2(I_e)}- \|f_e' \|^2_{L^2(I_e)}\right).
\end{equation}
For arbitrary $x\in P_-(\dom(L_v))$ we construct a function $f$ with the above properties, such that the norms $\|f\|_{L^2(X_E)}$, $\|f'\|_{L^2(X_E)}$ can easily be bounded by $\|x\| $.

We start by defining a constant $\ep:=\min\left\{ \frac{3}{5\,c},\ukl \right\}$.
For an edge starting in $v$ we set $f_e$ to be a piecewise linear function with
$f(0)=x_{(0,e)} $ and  $f'_e(0)=(L_v x)_{(0,e)}$ on the interval $(0,\delta)$ with $0<\delta < \frac{\ep}{4}$ and continue linearly, such that $f_e(\frac{\ep}{2})=0$. The function value stays zero afterwards.
On all edges with endpoint equal to $v$ we set an analogous function beginning at the endpoint.
On all edges not incident to $v$ the function is set equal to zero. Note that $f$ satisfies the boundary condition but only lies in $W^{1,2}(X_E)$.

For $a:=x_{(0,e)}$ and $b:=(L_v x)_{(0,e)}$ we chose $\delta$, such that $0<\delta<\min\left\{\frac{\ep}{4},\frac{|a|}{|b|},\frac{|a|^2}{|b|^2}  \right\}$ and estimate by direct calculation:
\begin{align*}
\|f\|^2_{L^2(X_E)}&\leq \frac{5}{3}\ep\,\|\tr{v}(f) \|^2, \qquad
\|f'\|^2_{L^2(X_E)}\leq \left(1+\frac{24}{\ep}\right)\|\tr{v}(f) \|^2.
\end{align*}
By smoothening $f$ in the non-differentiable points, we get $f\in W^{2,2}(X_E)$.
Then the last estimates together with \eqref{gl_Hpl_nubeschr} result in:
\begin{align*}
0\leq-\langle L^-_v \tr{v}(f),\tr{v}(f)\rangle  &\leq \sum\limits_{e \sim v} \left( c\,\| f\|^2_{L^2(I_e)}+ \|f' \|^2_{L^2(I_e)}\right)\\
&\leq \left(1+\frac{24}{\ep}+\frac{5}{3}\,c\,\ep\right)\, \| \tr{v}(f) \|^2\\
&\leq \left(2+40\,c+\frac{24}{\ukl}\right)\| \tr{v}(f) \|^2.
\end{align*}

\item[$\Longleftarrow$:]
Let $f\in \dom(H^{P,L})$. Then we get with remark \ref{bem_1}
\begin{align*}
\langle H^{P,L}f,f \rangle
&=\sum\limits_{v\in V} \langle L_v^+ \tr{v}(f),\tr{v}(f)\rangle+\sum\limits_{v\in V} \langle L_v^- \tr{v}(f),\tr{v}(f)\rangle+
\|f'\|^2_{L^2(X_E)}\\
&\geq \sum\limits_{v\in V} \langle L_v^- \tr{v}(f),\tr{v}(f)\rangle+\|f'\|^2_{L^2(X_E)}\\
&\geq -S\sum\limits_{v\in V}  \| \tr{v}(f)\|^2+\|f'\|^2_{L^2(X_E)}.
\end{align*}
With relation \eqref{gl_Sobolev_Graph} this yields
\begin{equation*}
\langle H^{P,L}f,f \rangle \geq
-2S\left( \frac{2}{\ep}\|f\|^2_{L^2}(X_E)+\ep \|f'\|^2_{L^2}(X_E)\right)   +\|f'\|^2_{L^2(X_E)}
\end{equation*}
for $\ep\leq \ukl$. If we choose $\ep$ with $1-2S\ep \geq 0 $, we obtain
\begin{equation*}
\langle H^{P,L}f,f \rangle \geq
-\frac{4S}{\ep}\|f\|^2_{L^2(X_E)}.\qedhere
\end{equation*}
\end{itemize}
\end{proof}

With the last theorem we have shown, that each self-adjoint operator $H^{P,L} $ of theorem \ref{satz_1}, which satisfies boundary conditions of the form \eqref{RB:PLS}, has an associated lower bounded quadratic form. In the following we will explicitly compute  the associated  a form  and show that it coincides with the form stated in \cite{Kuchment-04}, provided the conditions there are satisfied.

\begin{defn}
Let $\Gamma$ be a metric graph with \eqref{geom:u}.
Let a boundary condition of the form \eqref{RB:PLS} be given.
The direct sum of the operators $L_v$ will be denoted with $L$, i.\,e.
$L=\bigoplus\limits_{v \in V} L_v$ with
\begin{equation*}\dom(L)=\left\{x=(x_{v})_{v \in V} \in \bigoplus\limits_{v\in V} \rrv \text{ with } x_{v}\in \dom(L_v) \text{ and }
\sum\limits_{v\in V}  \| L_v x_v \|^2 <\infty \right\}.
\end{equation*}
\end{defn}

\begin{bem}
\label{bem_s_L_besch}
\begin{itemize}
\item The operator $L$ is densely defined in $\bigoplus (1-P_v)(\rrv)$, self-adjoint, lower bounded and has an associated form, which we denote by $s_L$, $\dom(s_L)=\dom(L^\frac{1}{2})\subset \bigoplus (1-P_v)(\rrv)$. We have:
\begin{equation*}
s_L[x,y]=\langle Lx,y \rangle = \sum\limits_{v\in V} \langle L_v x_v,y_v \rangle \qquad \text{for all } x \in \dom(L), y\in \dom(s_L).
\end{equation*}
\item The operator $L$ and therefore also its form $s_L$, are lower bounded with the same lower bound $-S$.
\end{itemize}
\end{bem}


\begin{defn}
\label{def_form}
Let $\Gamma$ be a metric graph with \eqref{geom:u}. For each boundary condition of the form \eqref{RB:PLS} we define a quadratic form
$\form_L$ by
\begin{align*}
\dom(\form_L)&=\{ f\in W^{1,2}(X_E)\with \tr{}(f)\in \dom(s_L)\},\\
\form_L[f]&=\|f'\|^2_{L^2(X_E)}+s_L[\tr{}(f)].
\end{align*}
\end{defn}

\begin{bem}
\begin{itemize}
\item The sesquilinear form associated to $\form_L$ is given by
\begin{equation*}
\form_L[f,g]=\langle f',g'\rangle+s_L[\tr{}(f),\tr{}(g)] \qquad f,g \in \dom(\form_L).
\end{equation*}

\item For $x =\tr{}(f)\in \dom(L)$ a short calculation shows
 \begin{equation}\label{gl_s_L_L}
 \langle Lx,x \rangle = s_L[x]=\sum\limits_{v\in V} \langle L_v x_v,x_v \rangle,
 \end{equation}

\item
If  $L$ is bounded (or, equivalently, if $(L_v)$  are uniformly bounded) one obtains
$\dom(L)=\dom(s_L)=\bigoplus (1-P_v)(\rrv)$ and
\begin{align}
\langle Lx,x \rangle &= s_L[x]=\sum\limits_{v\in V} \langle L_v x_v,x_v \rangle, \\
\notag \form_L[f]&=\|f'\|^2_{L^2(X_E)}+\sum\limits_{v\in V} \langle L_v \tr{v}(f),\tr{v}(f) \rangle,\\
\notag \dom(\form_L)&=\{f\in W^{1,2}(X_E) \with P_v \tr{v}(f)=0 \}.
\end{align}
Thus, in this situation  we get the  quadratic form from \cite{Kuchment-04}.

\item Let $f \in \dom(H^{P,L})$. Then the properties of the boundary condition and the Sobolev inequality yield:
\begin{align*}
\sum\limits_{v\in V}\|L_v \tr{v}(f)\|^2 &=\sum\limits_{v\in V} \|(1-P_v) \str{v}(f')\|^2 \leq \sum\limits_{v\in V} \|\str{v}(f')\|^2\\
&\leq 2\left(\frac{2}{u}+u\right)\|f'\|^2_{W^{1,2}(X_E)}< \infty.
\end{align*}
This means
\begin{equation}
\label{gl_op_trace_inL}
f \in \dom(H^{P,L}) \Rightarrow \tr{}(f) \in \dom(L).
\end{equation}
\end{itemize}
\end{bem}
For later applications we prove the following proposition, which shows: The form domain is stable under multiplication with $ W^{1,\infty}$-functions which are continuous in the vertices.
\begin{hsatz}
\label{hsatz:form:komst_stetig}
Let $\Gamma$ be a metric graph with \eqref{geom:u} and a boundary condition of the form \eqref{RB:PLS} be given. If $\varphi\in W^{1,\infty}(X_E)$ is a function, which is continuous in all vertices, then ${\varphi f \in \dom(\form_L)}$ for all $f\in \dom(\form_L)$.
\end{hsatz}
\begin{proof}
Clearly $(\varphi f) \in W^{1,2}(X_E)$. The boundary values of the function $\varphi f$ satisfy $\tr{v}(\varphi f)=\tr{v}(\varphi)\tr{v}(f)=c_v\cdot \tr{v}(f)$ for some complex constant $c_v$. Therefore it holds $\tr{v}(\varphi f)=c_v\cdot  \tr{v}(f)\in \dom(L_v)$.
\end{proof}


\begin{hsatz}
Let $\Gamma$ be a metric graph with \eqref{geom:u} and $\form_L$ the quadratic form corresponding to a boundary condition of the form \eqref{RB:PLS}. Then $\form_L$ is a lower bounded, densely defined and closed form.
\end{hsatz}
\begin{proof}
Obviously $\dom(\form_L)$ is dense in $L^2(X_E)$.
From remark \ref{bem_s_L_besch} we know for all $f\in \dom(\form_L)$
\begin{align}
\notag s_L[\tr{}(f)] &\underset{\hphantom{\eqref{gl_Sobolev_Graph}}}{\geq} -S\|\tr{}(f)\|^2 = -S\sum\limits_{v\in V} \|\tr{v}(f)\|^2_{\rrv},\\
\label{gl_L_v_H1}
&\underset{\eqref{gl_Sobolev_Graph}}{\geq} -2S\left( \frac{2}{\ep} \|f\|^2_{L^2(X_E)}+\ep \|f'\|^2_{L^2(X_E)} \right).
\end{align}
respectively:
\begin{equation*}
2S\ep \|f'\|_{L^2(X_E)}^2  +s_L[\tr{}(f)]  \rangle \geq -\frac{4S}{\ep} \|f\|_{L^2(X_E)}^2.
\end{equation*}
By choosing $\ep=\min\{u,\frac{1}{2S} \}$ we get
\begin{equation*}
\form_L[f]=\|f'\|_{L^2(X_E)}^2  +s_L[\tr{}(f)] \geq -\frac{4S}{\ep} \|f\|_{L^2(X_E)}^2.
\end{equation*}
With the lower bound we can define the form-norm by
\begin{equation*}
\form_{L,\alpha}[f]:=\form_L[f]+\alpha \|f\|^2=\|f'\|^2_{L^2(X_E)}+s_L[\tr{}(f)] +\alpha \|f\|^2_{L^2(X_E)}\geq \|f\|^2_{L^2(X_E)},
\end{equation*}
where $\alpha>\frac{4S}{\ep}+1$.
For the closedness of the form we have to prove that $(\dom(\form_L),\sqrt{\form_{L,\alpha}})$ is complete.
We note, that $\sqrt{\form_{L,\alpha}}$ is not equivalent to the $W^{1,2}$-norm any more.
From \eqref{gl_L_v_H1} we get:
\begin{equation*}
\form_{L,\alpha}[f]=\|f'\|^2+s_L[\tr{}(f)]+\alpha\|f\|^2
\geq \left(-\frac{4S}{\ep}+\alpha \right)\|f\|^2_{L^2(X_E)}+\left(-2S\ep+1 \right)\|f'\|^2_{L^2(X_E)}
\end{equation*}
for all $\ep\leq u$. If we pick $\ep$, s.\,t. $-2S\ep+1 \geq \frac{1}{2} $:
\begin{equation*}
\form_{L,\alpha}[f]\geq \frac{1}{2}\|f\|^2_{L^2(X_E)}+\frac{1}{2}\|f'\|^2_{L^2(X_E)}=\frac{1}{2}\|f\|^2_{W^{1,2}(X_E)}.
\end{equation*}

Let $(f_n)$ be a $\sqrt{\form_{L,\alpha}}$-Cauchy sequence in $\dom(\form_{L,\alpha})$. Then $f_n$ converges in $W^{1,2}(X_E)$ in the corresponding norm to a function $f\in W^{1,2}(X_E)$, as this space is closed.
It remains to show $f_n \xrightarrow{\sqrt{\form_{L,\alpha}}} f$.
\begin{itemize}
 \item
From convergence of $(f_n)$ in $W^{1,2}(X_E)$ and \eqref{gl_Sobolev_Graph} we obtain
$\tr{}(f_n)\to\tr{}(f)$.
\item
$f_n$ is a $\sqrt{\form_{L,\alpha}}$-Cauchy sequence, i.\,e. we have $\sqrt{\form_{L,\alpha}[f_n-f_m]}\to 0$ for $n$, $m\to \infty$ and
\begin{align*}
\form_{L,\alpha}[f_n-f_m]&=\|f'_n-f'_m\|^2 +s_L[\tr{}(f_n-f_m)]+\alpha \|f_n-f_m\|^2.
\end{align*}
Obviously $\sqrt{s_L[\tr{}(f_n-f_m)]}$ also converges to zero for $n,\;m \to \infty$. Hence $(\tr{}(f_n))$ is a Cauchy sequence in the norm induced by $s_L$.
As $s_L$ is a closed form, $\tr{}(f_n)$ converges to some $x\in \dom(s_L)$:
\begin{equation*}
\sqrt{s_L[\tr{}(f_n)-x]+\alpha \| f_n-x \| }\to 0.
\end{equation*}
Since $\sqrt{s_{L,\alpha}[\cdot]}\geq \|\cdot \|$ holds, we have $\| \tr{}(f_n)-x\|\to 0 $.
From convergence of $\tr{}(f_n)$ to $\tr{}(f)$ in the  $\ell^2$-norm, we get $x=\tr{}(f)$ and thus $f\in \dom(\form_L)$.\qedhere
\end{itemize}
\end{proof}


\begin{satz}
\label{satz:form:op}
Let $\Gamma$ be a metric graph with \eqref{geom:u} and $\form_L$ the quadratic form corresponding to the boundary condition  \eqref{RB:PLS} in the sense of definition \ref{def_form}. Then the self-adjoint operator associated to $\form_L$ is given by the operator $H^{P,L}$ of theorem \ref{satz_1}.
\end{satz}

The proof uses the same arguments as the proof in \cite{Kuchment-04}, where it was assumed that
$d_v<\infty $ and $\|L_v\|\leq S $ in all vertices $v \in V$.
Additionally to those arguments we have to show $f\in \dom(H^{P,L})\Rightarrow \sum\limits_{v\in V} \|L_v \tr{v}(f) \|^2 <\infty $, which follows by relation \eqref{gl_op_trace_inL}.

\begin{proof}

We denote the associated operator of $\form_L$ by $M_L$.
\begin{itemize}
\item Let $f\in \dom(H^{P,L})$. Then $f\in W^{1,2}(X_E)$, $\tr{v}(f)\in \dom(L_v)$ and $f\in \dom(\form_L)$ by \eqref{gl_op_trace_inL}.
The representation of the domain of $M_L$ is
\begin{equation*}
\dom(M_L)=\{h \in \dom(\form_L) \with  \esgibt g\in L^2(X_E) \text{ with } \langle g,\phi \rangle = \form_L[h,\phi] \fa \phi \in \dom(\form_L)\}.
\end{equation*}
Setting $g=H^{P,L}f=-f''$ we find with remark \ref{bem_1} and \eqref{gl_s_L_L} that $f\in \dom(M_L)$.
\item
Let $f\in \dom(M_L)$, i.\,e.  for all $g\in \dom(\form_L)$:
\begin{equation*}
\label{gl_assozop}
 \form_L[f,g]=\langle M_L f,g \rangle=\langle f',g' \rangle+s_L[\tr{}(f),\tr{}(g)].
\end{equation*}
If we insert test functions $\phi$ with compact support contained in one edge (yielding $\tr{}(\phi)=0$) we can conclude $f\in W^{2,2}(X_E)$.
By partial integration we find
\begin{align}
\label{gl_form_assozop} \langle M_L f,g \rangle&=-\langle f'',g \rangle+s_L[\tr{}(f),\tr{}(g)]-\sum\limits_{v \in V} \langle \str{v}(f'),\tr{v}(g) \rangle
\end{align}
for all functions $g\in \dom(\form_L)$. For all those functions $g$ exists a sequence $g_n\in \prod C^\infty\komp(I_e)\cap L^2(X_E) $ with $g_n \to g $ in $L^2(X_E)$. As $(g-g_n)\in \dom(\form_L)$ we can insert $(g-g_n)$ in \eqref{gl_form_assozop}, where we find with $\tr{}(g_n)=0$:
\begin{align*}
 \langle M_L f,g-g_n \rangle+\langle f'',g-g_n \rangle=s_L[\tr{}(f),\tr{}(g)]-\sum\limits_{v \in V} \langle \str{v}(f'),\tr{v}(g) \rangle.
\end{align*}
From convergence of $g_n$ to $g$ we see, that both sides have to be equal to zero.
All functions $g\in W^{1,2}(X_E)$ with $\tr{v}(g)\in \dom(L_v) $, which are only supported in a small neighborhood of the vertex $v$ are elements of $\dom(s_L) $. Inserting these in the RHS of the last equation, we get
\begin{equation*}
\langle L_v \tr{v}(g),\tr{v}(f)\rangle =\langle \tr{v}(g),\str{v}(f')\rangle,
\end{equation*}
from where we conclude $\tr{v}(f)\in \dom(L_v^*)=\dom(L_v) $ and $L_v \tr{v}(f)=(1-P_v)\str{v}(f')$, as it is true in a dense subset of $(1-P_v)(\rrv)$.\qedhere
\end{itemize}
\end{proof}

\section{Remarks on the situation with edge lengths tending  to zero}
\label{ab:KGN}

The main thrust of the paper is to remove various boundedness conditions imposed on the literature on quantum graphs. In this vein we have discussed a setting in which neither boundedness of the operators $L$ nor finiteness of the vertex degree is necessary. 
The only remaining restriction for metric graphs is the uniform bound on the lower edge lengths, which is crucial for most of the methods used in this paper.  If this uniform bound is missing, the situation gets substantially harder. In fact, only very  few results on selfadjointness of the Laplacian are known in that case. These include the  trivial ones:
\begin{enumerate}
\item
If the Laplace operator is decoupled in each edge, e.\,g. by Dirichlet or Neumann boundary condition, then the operator is self-adjoint on each single edge and as $\ell^2$-sum of self-adjoint operators a self-adjoint operator on the whole graph.
\item
Trivial boundary conditions (called free or Kirchhoff boundary condition) give a self-adjoint operator on $\rz$ or $\rz^+$ with local point interactions (i.\,e. vertices with Kirchhoff b.\,c.).
\end{enumerate}
As for the non-trivial ones there are  a  few works which treat the Laplace operator on the real axis, half axis or subsets of the axis with local point interactions, i.\,e. a metric graph, which can be imbedded in the real line.
\begin{enumerate}
 \item[3.]
Buschmann, Stolz and Weidmann proved in \cite{BuschmannSW}: The
Laplace operator with $\delta'$ interactions with arbitrary coupling constants on a discrete set in $\rz$ is self-adjoint.
The boundary conditions are encoded by unitary operators (see e.\,g. \cite{Harmer-00} for encoding b.\,c. with unitary operators on metric graphs).
\item[4.]
In \cite{KostenkoM-10} Kostenko and Malamud treat Laplace operators with
$\delta$ or $\delta'$ boundary conditions on $\rz$ or a part of $\rz$ with local point interactions.
They give necessary and sufficient conditions for selfadjointness by relations between the coupling constants $\alpha_n$ and the edge lengths $l_n$.
\end{enumerate}
For more general metric graphs, there is only one work known to the authors.
\begin{enumerate}
\item[5.]
Let a non-closed metric graph be given with the property that the
completion is compact and for every element of the boundary 
each nonempty, open neighborhood has  infinite volume.
Then Carlson proved in \cite{Carlson-08} that there exists a unique self-adjoint version of the Laplacian with a certain version of Kirchhoff boundary conditions.
\end{enumerate}

\begin{appendix}

\section{Lagrangian subspaces and boundary triplets}
\label{ab:LUR}

In this section we state a one-to-one correspondence between the boundary conditions of the form \eqref{RB:PL} and Lagrangian subspaces via the theory of  boundary triplets. This allows  one to  find all (as opposed to 'all local')  self-adjoint versions of the Laplacian by using Lagrangian subspaces. In our context it is of interest, as it provided the connection of our work to corresponding considerations, see e.\,g. \cite{KostrykinS-99b,Harmer-00,Post,SSVW}, and gives a possibility to prove a converse to theorem \ref{satz_1}.

\begin{defn}
Let $(\gil,\langle \cdot, \cdot \rangle)$ be a Hilbert space.
\begin{enumerate}
\item
The mapping $ \Omega :(\gil\oplus \gil) \times (\gil \oplus \gil) \to \K$  with
\begin{equation*}
\Omega(x,y)=\Omega((x_1,x_2),(y_1,y_2))= \langle x_2,y_1 \rangle -\langle x_1,y_2 \rangle
\end{equation*}
is a hermitian symplectic form. Let $S$ be the mapping $S:\gil\oplus \gil \to \gil\oplus \gil$ with $S(x_1,x_2)=(x_2,-x_1)$.
Then we have $\Omega(x,y)=\langle Sx,y \rangle =\langle x_2,y_1 \rangle+ \langle -x_1,y_2\rangle$.

\item
A subspace $G$ of the direct sum $\gil \oplus \gil $ is called linear relation.
The subspace
\begin{equation*}
G^*:=\{(x_1,x_2)\in \gil\oplus\gil \text{ with } \langle x_1,y_2\rangle=\langle x_2,y_1 \rangle \text{ for all } (y_1,y_2)\in G\}
\end{equation*}
is called the adjoint relation to $G$. If $G\subset G^*$, then $G$ is called symmetric and for
$G=G^*$ self-adjoint.

With the symplectic form we get $G^*=\{ x \in \gil\oplus\gil \text{ with } \Omega(x,y)=0 \fa y\in G  \} $.
\item
In our context a linear relation $G\subset \gil\oplus\gil$ is called Lagrangian subspace, if $(SG)^\bot=G $ holds.
\end{enumerate}
\end{defn}
For a more general definition of Lagrangian subspaces see \cite{Everitt}.

\begin{bem}
\label{bem_lag_subsp}
Let $\gil$ be a Hilbert space and $G\subset \gil \oplus \gil$. Then:
\begin{enumerate}
\item
The set $G$ is a Lagrangian subspace, iff $G$ is a self-adjoint linear relation.
\item
The set $G$ is a Lagrangian subspace, iff there is an orthogonal projection $P$ on a closed subspace of $\gil$ and a self-adjoint operator $L$ in $(1-P)\,\gil$, such that
\begin{equation*}
G=\{(q,Lq+p) \in \gil\oplus \gil \text{ with } Pp=p \text{ and } q\in\dom(L) \}.
\end{equation*}
\end{enumerate}
\end{bem}
\begin{proof}
The first part is clear and the second was proven in \cite{Arens} as theorem 5.3.
\end{proof}
Note that the last characterization exactly describes the boundary conditions of the form \eqref{RB:PL} (where $q$ equals $\tr{v}(f)$ or $\tr{}(f)$ and $Lq+p$ equals $\str{v}(f')$ or $\str{}(f')$).

Another way of characterizing Lagrangian subspaces is by unitary operators---see \cite{Harmer-00} and \cite{KostrykinS-00} for application on quantum graphs.

\begin{defn}
Let $H$ be a symmetric operator in the Hilbert space $\hil$, $\gil$ another Hilbert space and $F_1$, $F_2 : \dom(H^*)\to \gil$ be two linear functions, such that $(F_1,F_2):\dom(H^*)\oplus \dom(H^*)\to \gil \oplus \gil$ is surjective.
The tuple $(F_1,F_2, \gil)$ is called boundary triplet if the following condition is satisfied
\begin{align}
\label{gl_bt}
\langle f,H^*g \rangle-\langle H^*f,g \rangle = \langle F_1 (f),F_2(g)\rangle-\langle F_2 (f),F_1(g)\rangle \qquad \text{for all } f,g\in \dom(H^*).
\end{align}
\end{defn}
For general theory of boundary triplets see for instance \cite{DM,BrueningGP-08} and the references therein.

\begin{bsp}
\label{bsp_bt}
Let $\Gamma$ be a metric graph with \eqref{geom:u}.
The minimal Laplacian $\Delta_\mathrm{min}$ is defined on $\left(\prod\limits_{e \in E}C^\infty_c (I_e)\right)\cap W^{2,2}(X_E)$. Its adjoint is the maximal Laplacian, which is defined on $W^{2,2}(X_E)$.
We set $\gil:=\bigoplus\limits_{v \in V} \rrv$, $F_1(f):=\tr{}(f)$ and $F_2(f):=\str{}(f')$ for all $f\in W^{2,2}(X_E)$, which are clearly linear and $(F_1,F_2)$ is surjective since remark \ref{bem_rrv2}.
Relation \eqref{gl_bt} follows from proposition \ref{hsatz_0}. Thus we have a boundary triplet  $(F_1,F_2,\gil)$ for the Laplacian.
\end{bsp}

Now we can apply the known result that all self-adjoint extensions of a symmetric operator can be found via a boundary triplet and Lagrangian subspaces (see e.\,g. theorem 1.12 in \cite{BrueningGP-08}) to quantum graphs.

\begin{satz}
\label{satz_HG}
Let $\Gamma$ be a metric graph with \eqref{geom:u}.
The operator $H_G$ is self-adjoint, iff $G$ is a Lagrangian subspace of $\bigoplus\limits_{v\in V} \rrv\times \bigoplus\limits_{v\in V} \rrv$.
Here $H_G$ is defined by
\begin{align*}
\dom(H_G)&=\{ W^{2,2}(X_E) \text{ with } (\tr{}(f),\str{}(f'))\in G \}\\
H_Gf &= -f''
\end{align*}
\end{satz}
With example \ref{bsp_bt} this theorem is a corollary of theorem 1.12 in \cite{BrueningGP-08}, see also \cite{SSVW} for a notation with boundary conditions in the form \eqref{RB:PL} and characterizations of first derivative operators on metric graphs by boundary systems.

In the last theorem the Lagrangian subspace and thus the boundary condition doesn't see
the structure of the graph  (which means the Lagrangian subspace might not be decomposable corresponding to the vertex structure).

If we restrict the boundary condition to local interactions in each vertex, which is sometimes called vertex boundary conditions, we get the same result. This means the operator is self-adjoint if and only if in each vertex a Lagrangian subspace is chosen for the boundary condition.
\begin{satz}
\label{satz_Hsa_LUR}
Let $\Gamma$ be a metric graph with \eqref{geom:u}.
The operator $H_G$ is self-adjoint if and only if the boundary condition in each vertex is defined by a Lagrangian subspace.
Here $H_G$ is given by
\begin{align*}
\dom(H_G)&=\{f\in W^{2,2}(X_E) \text{ with }  (\tr{v}(f),\str{v}(f'))\in G_v \fa v\in V \},\\
H_G f&=-f''.
\end{align*}
\end{satz}

\begin{proof}
One direction follows by remark \ref{bem_lag_subsp} point 2. and theorem \ref{satz_1}.
Both directions where directly proven in theorem 1.5.7 in \cite{diss_11}. The second direction also follows by theorem \ref{satz_HG} and the fact, that the direct product of Lagrangian subspaces is a Lagrangian subspace.
\end{proof}

\end{appendix}

\printbibliography
\bigskip
\noindent
\begin{tabular}{l}
Daniel Lenz\\
Fakult\"at f\"ur Mathematik und Informatik \\
Friedrich-Schiller-Universit\"at Jena \\
Ernst-Abbe-Platz 2, 07743 Jena, Germany\\
{\tt daniel.lenz@uni-jena.de}
\end{tabular}\\[3ex]
\begin{tabular}{lll}
Carsten Schubert & & and\\
Fakult\"at Mathematik & &  Fakult\"at f\"ur Mathematik und Informatik \\
Technische Universit\"at Chemnitz& &  Friedrich-Schiller-Universit\"at Jena \\
09107 Chemnitz, Germany  & & Ernst-Abbe-Platz 2, 07743 Jena, Germany\\
\end{tabular}\\
\begin{tabular}{l}
{\tt carsten.schubert@mathematik.tu-chemnitz.de}
\end{tabular}\\[3ex]
\begin{tabular}{l}
Ivan Veseli\'c\\
Fakult\"at Mathematik\\
Technische Universit\"at Chemnitz\\
09107 Chemnitz, Germany \\
{\tt ivan.veselic@mathematik.tu-chemnitz.de}
\end{tabular}\\[3ex]

\end{document}